\documentclass{article}
\usepackage{arxiv}
\usepackage{placeins}
\usepackage{algpseudocode}
\usepackage{soul}
\usepackage{tikz}
\usepackage{float}
\usepackage{xspace}
\usepackage{xcolor}
\usepackage{amsthm}
\usepackage{amsmath}
\newtheorem{definition}{Definition}
\newtheorem{theorem}{Theorem}

\usepackage{amssymb}
\usepackage{graphicx}
\usepackage{longtable}
\usepackage[linesnumbered,ruled,vlined]{algorithm2e}
\usepackage{placeins}
\usepackage{hyperref}       
\usepackage{url}            
\usepackage{booktabs}       
\usepackage{amsfonts}       
\usepackage{nicefrac}       
\usepackage{microtype}      
\usepackage{lipsum}
\usepackage{graphicx}
\graphicspath{ {./images/} }
\usepackage{placeins}
\usepackage{xcolor}
\usepackage{comment}
\usepackage{caption}
\usepackage{soul}
\usepackage{tikz}
\usepackage{float}
\usepackage{xspace}
\usepackage{amsmath}
\usepackage{amssymb}
\usepackage{graphicx}
\usepackage{longtable}
\usepackage[linesnumbered,ruled,vlined]{algorithm2e}
\usepackage{stmaryrd} 

\newcommand{\Real}{\ensuremath{\mathbb{R}}}
\newcommand{\Integer}{\ensuremath{\mathbb{Z}}}
\newcommand{\Ratio}{\ensuremath{\mathbb{Q}}}

\newcommand{\feature}{\ensuremath{\xi}}
\newcommand{\particle}{\ensuremath{\Xi}}
\newcommand{\emptysymbol}{\ensuremath{\square}}
\newcommand{\air}{\text{space}}
\newcommand{\tdv}[2]{\ensuremath{\begin{pmatrix}#1 \\ #2\end{pmatrix}}}
\newcommand{\tdm}[4]{\ensuremath{\begin{pmatrix}#1 & #2 \\ #3 & #4\end{pmatrix}}}

\newcounter{myobsctr}

\numberwithin{myobsctr}{section}

\newcounter{myprectr}

\newenvironment{proposition}{
	\bigskip\noindent
	\refstepcounter{myprectr}
	\textbf{Proposition \themyprectr}
	\newline%
}{\par\bigskip}  
\numberwithin{myprectr}{section}

\usepackage[utf8]{inputenc} 
\usepackage[T1]{fontenc}    
\usepackage{hyperref}       
\usepackage{url}            
\usepackage{booktabs}       
\usepackage{amsfonts}       
\usepackage{nicefrac}       
\usepackage{microtype}      
\usepackage{lipsum}		
\usepackage{graphicx}
\usepackage{doi}
\usepackage[square,sort,comma,numbers]{natbib}

\title{The 2D Ray Tracing Problem using ABCD Lenses and Mirrors is Turing Complete}

\author{ \href{https://orcid.org/0009-0009-3654-891X}{\includegraphics[scale=0.06]{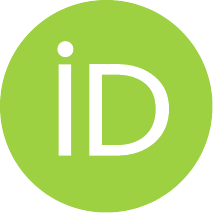}\hspace{1mm}Rosemary U. Adejoh}\\
	Faculty of Media\\
	Bauhaus University\\
	Weimar, Germany \\
	\texttt{rosemary.utenwojo.adejoh@uni-weimar.de} \\
	\And
	\href{https://orcid.org/0000-0002-9989-7801}{\includegraphics[scale=0.06]{orcid.pdf}\hspace{1mm}Andreas Jakoby} \\
	Faculty of Media\\
	Bauhaus University\\
	Weimar, Germany \\
	\texttt{andreas.jakoby@uni-weimar.de} \\
	\And
	\href{https://orcid.org/0000-0001-7881-4563}{\includegraphics[scale=0.06]{orcid.pdf}\hspace{1mm}Sneha Mohanty} \\
	Computer Networks and Telematics\\
	University of Freiburg \\
	Freiburg, Germany \\
	\texttt{mohanty@informatik.uni-freiburg.de}\\
	\And
	\href{https://orcid.org/0000-0002-8320-8581}{\includegraphics[scale=0.06]{orcid.pdf}\hspace{1mm}Christian Schindelhauer} \\
	Computer Networks and Telematics\\
	University of Freiburg \\
	Freiburg, Germany \\
	\texttt{schindel@informatik.uni-freiburg.de}}	

\renewcommand{\headeright}{Technical Report}
\renewcommand{\undertitle}{Technical Report}
\renewcommand{\shorttitle}{The 2D Ray Tracing Problem using ABCD Lenses and Mirrors}

\hypersetup{
pdftitle={A template for the arxiv style},
pdfsubject={q-bio.NC, q-bio.QM},
pdfauthor={David S.~Hippocampus, Elias D.~Striatum},
pdfkeywords={First keyword, Second keyword, More},
}

\begin{document}
\maketitle

\begin{abstract}
We establish that the two-dimensional ray tracing problem with thin lenses and plane mirrors is Turing-complete, thereby resolving an open question posed by Reif et al. in 1994 as to whether three-dimensional space is necessary for computational universality in optical systems. To this end, we consider the standard approximation of reflection and refraction, namely the ABCD model for paraxial optics, which describes ray propagation through lenses (refraction) via a $2 \times 2$ matrix, combined with the geometric reflection model for plane mirrors.

In the absence of mirrors, two-dimensional ray tracing using any combination of lenses in this ABCD matrix model can be described by a single $2 \times 2$ matrix–vector product, where the matrix has real entries and determinant~1. Conversely, we show that any such matrix with determinant~1 can be represented as a composition of exactly three appropriately spaced thin lenses.
When mirrors are combined with lenses, the ray tracing problem can be described by a flowchart using only two variables, which establishes Turing computability for rational-valued inputs, spaces and matrix entries. Building on this observation, we present a construction of ray tracing that simulates a reversible Turing machine.

We begin with a restricted version of the reversible flowchart problem, in which only two variables and certain linear functions are permitted. We prove that this restricted variant is Turing-complete. We then show that such a flowchart admits a geometric realization using lenses and mirrors in our model, thereby establishing the main result: Turing-completeness of the two-dimensional ray tracing problem with ABCD-model lenses and mirrors.
%
%
%
%
\end{abstract}


\keywords{Turing completeness, optical computation, ray tracing, ABCD matrix, thin lenses, plane mirrors, reversible Turing machine, flowchart, reversible flowchart}

\section{Introduction and Motivation}

\label{intro}

The computational power of physical systems is a topic of sustained interest in Theoretical Computer Science. Understanding which physical models can simulate universal computation not only showcases the boundaries between tractable and intractable problems but also reveals deep connections between physics and computation.
Among such models, some of the most elegant ones are based on optics, where the physics of light propagation through lenses as well as reflections from mirrors give rise to intuitively rich computational behavior. 

Here, we concentrate on the computational complexity of ray tracing in 2D optical systems, where a light ray enters a maze of lenses and mirrors and we are concerned with the question of where the ray leaves this system.
Our main tool is the ABCD ray transfer matrix, which originates from paraxial optics, describing the propagation of light rays through optical elements. 

In this model, a ray is characterized by two parameters, i.e. its height offset $y$ (the perpendicular distance from a reference axis) and its angular offset $\theta$, 
the angle relative to that axis. The propagation of a ray through any optical element—whether a thin lens or a stretch of empty space—is described by multiplication by a $2 \times 2$ real matrix: 
\begin{equation}\label{eq_matrix}
	\begin{pmatrix} y_1 \\ \theta_1 \end{pmatrix} = \begin{pmatrix} A & B \\ C & D \end{pmatrix} \cdot \begin{pmatrix} y_0 \\ \theta_0 \end{pmatrix}
\end{equation}
where the matrix has unit determinant ($AD - BC = 1$), i.e. the ABCD matrix is in the special linear group $\mathrm{SL}(2, \mathbb{R})$, which is for $2 \times 2$  matrices symplectic \cite{Bastiaans2007ClassificationOL}. This ABCD matrix formalism, introduced by Kogelnik \cite{kogelnik1965imaging} building on work of Fresnel, is a standard tool in optical engineering for designing telescopes, microscopes, and laser cavities. We demonstrate that it is also a natural framework for universal computation.

\subsection{Prior Work and Dimensionality Question}
The question of whether optical or mechanical reflection systems can perform universal computation has a rich history. In their seminal paper, Reif et al. \cite{Reif1994ComputabilityAC} proved that three-dimensional ray tracing with mirrors is Turing-complete by showing how to simulate a reversible two-counter automaton. Their construction relies heavily on parabolic mirrors in 3D space to implement the counter operations. Prior to that, Moore \cite{moore1990unpredictability} established similar results for frictionless billiard ball systems, showing that a single ball in a 3D maze of mirrors can simulate a two-counter machine.

It remains an open question whether three spatial dimensions are necessary for Turing completeness, or do two-dimensional ray tracing suffice.

Recently in \cite{adejoh_fsttcs_final}, we made progress on this question by introducing the ray particle tracing problem, a two-dimensional model involving a ray particle navigating through stationary walls, moving walls, and one-way gates. We showed that this model can simulate a two-stack pushdown automaton (equivalent to a Turing machine) by encoding one stack in the spatial offset of the ray particle's position and the other stack in the temporal offset of its arrival time. While this establishes Turing-completeness in a geometric sense, the model is not purely two-dimensional, i.e. the use of time offset as a computational resource effectively introduces a third parameter, making it a 2.5-dimensional construction. Furthermore, the reliance on moving walls and one-way gates places the model outside the domain of classical static optics.

\subsection{Our Contribution}
We resolve the dimensionality question discussed above, i.e.
we show that the Turing-completeness of the 2D ray tracing problem can be realised using only thin lenses and plane mirrors, thereby getting rid of moving walls (time offset manipulation related components) as well as one-way gates previously used in \cite{adejoh_fsttcs_final}. This means that only static optical components suffice for our constructions. Also, instead of encoding stacks in spatial offset and time offset, we encode the left and right sides of a 1-tape RTM using the height offset $y$ and the angular offset $\theta$ of the ray.  

We first investigate the computational power of thin lenses in the ABCD model and show that every matrix with determinant~1 can be realised using a constant number of lenses. By adding mirrors to the model, the ray tracing problem can be expressed as a reversible flowchart that operates on only two variables and uses solely linear operations, such as; \emph{Step}, \emph{Test}, and \emph{Assertion}. This establishes computability for rational-valued inputs.
We then show that such restricted reversible flowcharts are reversible Turing complete. This serves as the main tool for proving our central result: that 2D ray tracing is Turing-complete.

The remainder of our paper is organized as follows:
Section~\ref{rwork} presents related work on computational models, the reversible flowchart problem as well as ray/beam propagation through complex optical systems.
Section~\ref{tracing} establishes our problem definition.
Section~\ref{TC_system} introduces the abstract components required for Turing-completeness and Theorem~\ref{th:main-thm_rt}, our main result on the Turing-completeness of the 2D ray tracing problem, as well as the Turing-completeness of the related two-variable linear reversible flowchart.
The proof for this flowchart problem is given in Section~\ref{sec:Red-2-LRFC}.
Section~\ref{abcd} develops the ABCD matrix formalism and our proof that every ABCD matrix can be constructed with a constant number of thin lenses of the same focal length $f > 0$.
Finally, Section~\ref{lens-tracing} combines all components and elaborates on the main result on the Turing-completeness of 2D ray tracing.
Section~\ref{conclusions} contains the conclusion and open problems.
The appendix contains all omitted proofs for Section~\ref{abcd} as well as some figures moved there due to space limitations.

\section{Related Work}
\label{rwork}

The theory of reversible computation plays an important role in this paper, which was investigated by
Toffoli in \cite{toffoli} based on invertible primitives and invertibility preserving composition rules. Given such constraints, a closer correspondence between the behavior of abstract computing systems and underlying microscopic physical laws can be attained. One of the tools we use in our main result are reversible flowcharts, the fundamentals of  which are presented by 
Yokoyama, Axelsen and Glück in \cite{ReversibleFlowchart}. They show that reversible flowcharts, in its structured and unstructured form are reversible Turing complete, implying that they can compute exactly all injective computable functions. In \cite{AxelsenBock2011}, Axelsen and Glück discuss that Reversible computing is the study of computation models
that illustrate both forward and backward determinism. They stress that the property of reversibility should be the starting point of
a computational theory of reversible computing. They show that the RTMs can compute exactly all injective, computable functions and that r-Turing completeness is
the gold standard for computability in reversible computation models.
 Chardonnet, Lemonnier and Valiron in \cite{chardonnet_2024} show how one can encode a reversible Turing machine using the Theseus language. They build a robust categorical semantics based on join inverse
categories and present additional structures to capture pattern-matching and to interpret inductive types
as well as recursion.

The work that is closest to our findings is by Miranda and Ramos \cite{miranda2025}, who show that two-dimensional billiard systems are Turing complete by encoding their dynamics within the framework of Topological Kleene Field Theory. They construct a smooth billiard wall with a fractal zig--zag profile. Their results establish the existence of undecidable trajectories in physically natural two-dimensional billiard-type models. González-Prieto, Miranda and Peralta-Salas \cite{González-Prieto_2025} present a survey in which they review recent work introducing novel perspectives on the representation of computability through dynamical systems, ranging from classical dynamical universality to modern notions of Turing universality in fluid dynamics and Topological Kleene field theories.

While the ABCD model is only an approximation, it is still relevant in modern optical physics.
In \cite{Bastiaans2007ClassificationOL}, Bastiaans and Alieva propose a classification of first-order optical systems based on the eigenvalues of the ray transformation (ABCD) matrix. They present a classification of 1 and 2-dimensional real symplectic matrices based on the distribution of eigenvalues and on if the matrix can be diagonalized. 
In \cite{theodore_abcd}, Corcovilos presents an improved $3 \times 3$ matrix method for calculating paraxial ray traces of optical systems that is applicable to arrangements on an optical table, including lenses and mirrors in arbitrary orientation and position. Yura and Hanson \cite{Yura:87} describe a formulation of light beam propagation through complex optical systems, including effects such as finite apertures, random jitter, tilt as well as atmospheric turbulence. 



\section{Problem Definition}\label{tracing}

We are concerned with the ray tracing problem as depicted in Figure~\ref{f:problem-task}. 
\begin{definition}[2D ray tracing problem with ABCD lenses and mirrors]
Given a set of mirrors and lenses in two dimensions and a finishing line, the task is to determine whether a ray, given by its initial position and direction, passes the finishing line after being refracted at the lenses according to the ABCD model and reflected at mirrors according to the standard physical law of reflection.
\end{definition}

The law of reflection states that the angle of incidence equals the angle of reflection at the point of contact of the light ray with the mirror.
\begin{figure}
\begin{minipage}[b]{0.52\textwidth}
	\includegraphics[width=1\textwidth]{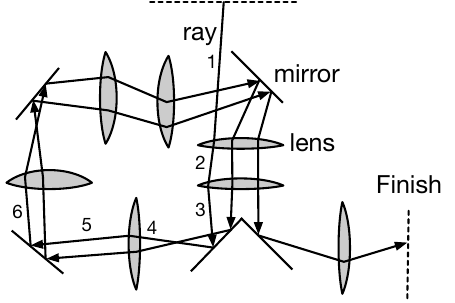}
	\caption{The 2D Ray Tracing problem with lenses and mirrors. The transformations steps of the ray are numbered. The remaining steps follow  a similar sequence.\label{f:problem-task}}
	\end{minipage}
\hfill
\begin{minipage}[b]{0.4\textwidth}
\includegraphics[width=1\textwidth]{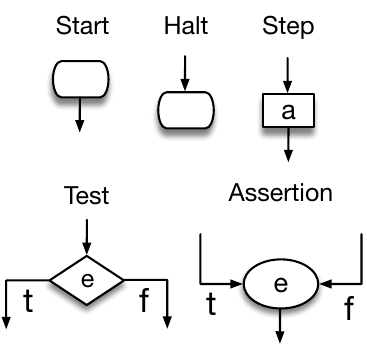}
\caption{The Reversible Flow Chart Components\label{flowchart}}
\end{minipage}
\end{figure}
%
%
According to the paraxial optics model, we specify the ray by its initial height offset $y_0 \in \mathbb{Q}$ and angular offset $\theta_0 \in \mathbb{Q}$ relative to the optical axis. Note that free space itself influences these parameters, as we will discuss later in Section~\ref{abcd}.

%
	
For our investigation, we use two reversible computation models. The first is the reversible flowchart model~\cite{ReversibleFlowchart}. We use this model to establish Turing computability of all parameters (of the initial ray, the lenses and the mirrors), assuming they are given as rational numbers.
%
%
Secondly, for the Turing hardness proof we use the reversible Turing machine in quadruple form as discussed in~\cite{morita2017reversiblecomputing}.

\section{Turing Complete Systems}
\label{TC_system}

Recall that a computational model is Turing-complete if it can simulate any given Turing machine. Well-known Turing-complete models include reversible Turing machines, two-stack pushdown automata, and two-counter automata.

In \cite{adejoh_fsttcs_final} we have shown that the Pinball Wizard problem is Turing-complete. Our simulation uses a network of gadgets which manipulates the arrival time and the arrival position of a ball within a pinball machine and by this we show the Turing-completeness of this problem. 

We extend this approach by emphasizing the close relationship to the reversible flowchart \cite{ReversibleFlowchart}.
 {\em  A reversible flowchart $F$} is a finite directed graph with three standard types of nodes: {\em Step}, {\em Test} and {\em Assertion}. Here, we have added the single {\em Start} node and a {\em Halt} node, indicating the begin and the end of the calculation, see Figure~\ref{flowchart}.

According to this figure the in- and out-degrees of these nodes are fixed. A further requirement is that the calculation done in a {\em Step} node, denoted as $a$, must be  invertible, denoted by $a^{-1}$.  In a {\em Test} node, a condition $e$ is applied in order to determine, which path the calculation takes after this node. 

The introduction of the {\em Assertion} (Join) node is a clever mechanism to ensure invertibility (backward determinism). The designer of a reversible flowchart must ensure that any calculation entering the Assertion node on the input marked as {\tt true} fulfills the condition given by $e$, while any calculation entering at {\tt false} fulfills the negation $\neg e$. 

Similarly, the inverted flowchart can be constructed by inverting the directions of all edges, exchanging Start and Halt nodes, replacing Test with Assertion nodes (and vice versa), and finally replacing the calculation $a$ on the Step nodes with $a^{-1}$, e.g. the reversal of Figure~\ref{f:leftheadmovement} into Figure \ref{f:rightheadmovement}.

In \cite{ReversibleFlowchart} it has been established that reversible flowcharts with a very basic command set for the Step operations can compute all reversible computations. Here, we use the unstructured reversible flowchart definition, which has been shown to be as expressive as the structured version according to \cite{ReversibleFlowchart}.

For analyzing ray tracing we introduce a reduced version of this model.
\begin{definition}
The 2-Variable Linear Reversible Flowchart [2-LRFC] is a reversible flowchart with the following restrictions
\begin{enumerate}
\item There are only two variables $y,\theta  \in \Ratio$.
\item In each step the calculations $a$ is a linear transformation, i.e.
$$ \begin{pmatrix} y \\ \theta \end{pmatrix} \leftarrow \begin{pmatrix} A & B  \\ C & D \end{pmatrix}  \begin{pmatrix} y \\ \theta \end{pmatrix} +  \begin{pmatrix} c_1 \\  c_2 \end{pmatrix} \ , $$
where the matrix is invertible, i.e. 
($AD - BC \neq 0$), and $A,B,C,D,c_1,c_2 \in \Ratio$.
\item The condition $e$ on the Test and Assertions nodes is linear, i.e. 
$ a y + b \theta  \leq c $, where $a,b,c \in \Ratio$.
\end{enumerate}
\end{definition}

Clearly, this model describes only reversible computable functions as the original reversible flowchart model.  We show that this model is reversible Turing-complete by simulating a reversible Turing machine in Section~\ref{sec:Red-2-LRFC}.
%
%
%
\begin{theorem}\label{t:2LRC-Complete}
2-LRFC is reversible Turing-Complete
\end{theorem}

For this we give an explicit construction of a reversible Turing machine (RTM), where the Step operations and Test/Assertion conditions are even further reduced. This reduced definition reflects the fact that for ABCD ray tracing the determinant fulfills $AD-BC=1$ and that mirrors can facilitate  Test and Assertion, if we only test for the first variable $y$. 
\begin{definition}\label{def:red2LRFC}
The Reduced 2-Variable Linear Reversible Flowchart [Red-2-LRFC] is a 2-LRFC with the following restrictions. We only use condition $y<1$ for the Tests and Assertions and the Step nodes facilitate only one of these calculations:
\begin{enumerate}
\item Shifting up by one unit: 
 $ \begin{pmatrix} y \\ \theta \end{pmatrix} =  \begin{pmatrix} y \\ \theta \end{pmatrix} +  \begin{pmatrix} 1\\ 0 \end{pmatrix} \  $
\item Shifting down by one unit: 
 $ \begin{pmatrix} y \\ \theta \end{pmatrix} =  \begin{pmatrix} y \\ \theta \end{pmatrix} - \begin{pmatrix} 1\\ 0 \end{pmatrix} \  $
	\item {\em Multiplication}: $ \begin{pmatrix} y \\ \theta \end{pmatrix} = M_\text{mult}  \begin{pmatrix} y \\ \theta \end{pmatrix}$, for $ M_\text{mult} =  \begin{pmatrix} 2 & 0 \\ 0 & \frac{1}{2} \end{pmatrix}$
	\item {\em Switch}:$ \begin{pmatrix}
	 y \\ \theta \end{pmatrix} =  M_\text{switch} \begin{pmatrix} y \\ \theta \end{pmatrix}  $, for $ M_\text{switch} =  \begin{pmatrix} 0 & 1 \\ -1 & 0 \end{pmatrix}$
	\end{enumerate}
	\end{definition}

We show that Red-2-LRFC is reversible Turing complete by giving an explicit construction of a simulation of a reversible Turing Machine (RTM) in Section~\ref{sec:Red-2-LRFC}, which then implies Theorem~\ref{t:2LRC-Complete}.

Returning to the ray tracing problem in two-dimensional paraxial optics, the variable $y$ denotes the spatial offset from the optical axis of the lenses, and $\theta$ denotes the angle of the ray. Neither $M_{\text{switch}}$ nor $M_{\text{mult}}$ can be directly realized by a single interaction of a ray with a lens or a mirror. Furthermore, while the ray travels between these elements, its offset $y$ changes, while $\theta$ remains constant. To familiarize the reader with the basics of this ABCD optical theory, we provide a short introduction in Section~\ref{abcd}, where we also discuss which matrix operations can be realized by a constant sequence of lenses without mirrors.

This enables us, in Section~\ref{lens-tracing}, to give an explicit construction of lenses and mirrors that transforms the reversible flowchart simulating an RTM from Section~\ref{sec:Red-2-LRFC} into an instance of the two-dimensional ray tracing problem. For this, we place Test and Assertion operations at the nodes of a carefully spaced grid. We then show how the four step operations, together with a connection element implementing the identity function, can be embedded along the edges of this grid. This establishes our main result.

%
%

\begin{theorem}\label{th:main-thm_rt}
	The two-dimensional ray tracing problem with ABCD lenses and plane mirrors is Turing-complete. 
\end{theorem}

\section{Turing Hardness of the Two-Variable Linear Reversible Flowchart}
\label{sec:Red-2-LRFC}

For our simulation, we use the quadruple form of a reversible Turing machine (RTM) as discussed in~\cite{morita2017reversiblecomputing}.

\begin{definition}[Reversible Turing Machine]
A one-tape reversible Turing machine in quadruple form is defined by
$M=(Q,\Sigma,q_0,q_f,\delta)$, where
$Q$ is a finite set of states, $\Sigma=\{0,1,\beta\}$ is the tape alphabet, $q_0 \in Q$ is the initial state, and $q_f \in Q$ is the final state.
The transition relation is given by
\[
\delta \subseteq (Q \times \Sigma \times \Sigma \times Q)\ \cup\ (Q \times \{\emptysymbol\} \times \{-1,0,1\} \times Q).
\]
A transition is either:
\begin{itemize}
    \item a \textbf{computing transition}, i.e., a transition $[p, x, z, q]$ with states $p,q \in Q$ and symbols $x,z \in \Sigma$. In this case, the head does not move; it reads symbol $x$ and writes symbol $z$.
    \item a \textbf{head-moving transition}, i.e., a transition $[p, \emptysymbol, d, q]$ with states $p,q \in Q$, a special symbol $\emptysymbol \notin \Sigma$ (indicating that no symbol is read), and a movement direction $d \in \{-1,0,+1\}$. In this case, the head moves left ($-1$), right ($+1$), or stays in place ($0$), and the tape remains unchanged.
\end{itemize}
To ensure reversibility, the transition relation $\delta$ must be injective in both forward and backward directions. That is, for any two distinct transitions 
$[p_1, x_1, y_1, q_1] \in \delta$ and $[p_2, x_2, y_2, q_2] \in \delta$, the following conditions hold:
\begin{itemize}
    \item \textbf{Forward determinism:} if $p_1 = p_2$, then $x_1 \neq x_2$.
    \item \textbf{Backward determinism:} if $q_1 = q_2$, then $y_1 \neq y_2$.
\end{itemize}
\end{definition}
We assume that within a computation the RTM always alternates from step to step between computing and head moving transition. 

In the following construction, we ensure that the two flowchart variables are always bounded, i.e., $\theta, y \in \Ratio \cap (-2,2)$. This is motivated by the limited size of the lenses used later on. Furthermore, we avoid integer values, i.e., $\theta, y \notin \Integer$, since this prevents rays from hitting the boundaries of mirrors.

We now describe how to construct a reduced 2-variable linear reversible flowchart for a given reversible Turing machine $M$. The initial and final states are represented by the Start node and Halt node of the flowchart, respectively. For the remaining states of the Turing machine, we introduce $|Q|-2$ state nodes in the flowchart; see Figure~\ref{overview}.

\begin{figure}
\begin{minipage}[b]{0.4\textwidth}
\includegraphics[width=\textwidth]{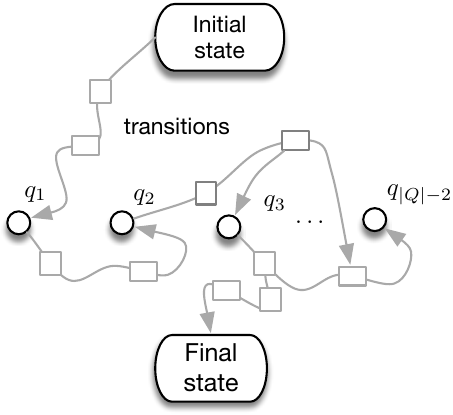}
\caption{State nodes and transitions of the reversible flowchart simulating an RTM. \label{overview}}
\end{minipage}\hfill
\begin{minipage}[b]{0.4\textwidth}
\includegraphics[width=\textwidth]{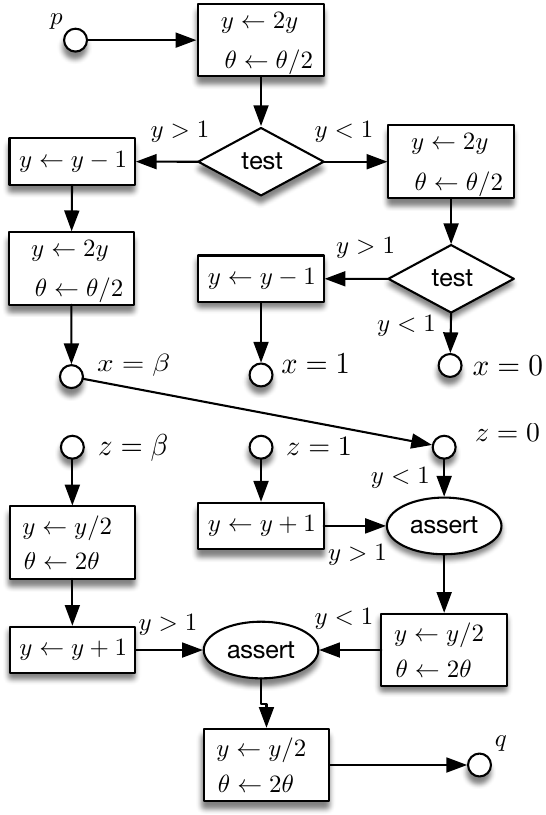}
\caption{Reversible flowchart for the computation transition $(p,x,z,q)$ for $x=\beta$ and $z=0$. \label{f:computationtransition}}
\end{minipage}
\end{figure}

The tape content and head position are encoded by the two variables $y$ and $\theta$, which assume consistent values upon entering each state node. The left part of the tape (relative to the head position) is encoded by $\theta$, while the symbol under the head together with the right part of the tape is encoded by $y$. Let
\[
\ldots s_{-3}\, s_{-2}\, s_{-1}\, s_0\, s_1\, s_2 \ldots
\]
be the infinite tape in some state $q$, where $s_0$ denotes the symbol under the head. Positive indices correspond to the right side of the tape, and negative indices to the left side. We encode each symbol $s \in \{0,1,\beta\}$ by an integer from $\{0,1,2\}$ via $v(0)=0$, $v(1)=1$, and $v(\beta)=2$. The tape is then represented as follows:
\begin{eqnarray}
y & = & \sum_{i=0}^{\infty} v(s_i)\cdot 4^{-i-1},\\
\theta & = & \sum_{i=1}^{\infty} v(s_{-i})\cdot 4^{-i}.
\end{eqnarray}

Note that an empty tape $\ldots \beta \beta \beta \ldots$ corresponds to the constant values $y=\theta=\frac{2}{3}$. In particular, neither $y$ nor $\theta$ ever attains an integer value. Using this encoding, we can realize the following elementary operations:
\begin{enumerate}
    \item Reading the symbol $s_0$ under the head: $v(s_0) = \lfloor 4y \rfloor$.
    \item Reading the symbol $s_{-1}$ to the left of the head: $v(s_{-1}) = \lfloor 4\theta \rfloor$.
    \item Replacing the symbol $s_0$ by $s_1$ at the head position:
    \(
    y \leftarrow y + \frac{v(s_1) - v(s_0)}{4}.
    \)
    \item Moving the head to the right:
    \begin{eqnarray}
        y &\leftarrow& 4y - v(s_0),\\
        \theta &\leftarrow& \frac{\theta + v(s_0)}{4}.
    \end{eqnarray}
    \item Moving the head to the left:
    \begin{eqnarray}
        y &\leftarrow& \frac{y + v(s_{-1})}{4},\\
        \theta &\leftarrow& 4\theta - v(s_{-1}).
    \end{eqnarray}
\end{enumerate}

We connect the state nodes of the reversible flowchart using components that implement the transitions given by $\delta$ of the RTM.

Given a {\em computation transition} $[p, x, z, q]$, we construct a sub-flowchart connecting state node $p$ to state node $q$. Let $v(x) = b_0 + 2b_1$ and $v(z) = b'_0 + 2b'_1$, where $b_i, b'_i \in \{0,1\}$.

We build a depth-two decision tree. First, we apply a multiplication step to test the most significant bit $b_1$ of $v(x)$. If $b_1=1$, we remove it by a suitable shift. We then test for $b_0$ in the second layer. After these steps, the symbol $x$ is effectively erased from the head position.

To write the new symbol $z$, we use the division operation
\[
M_{\text{div}} := M_{\text{switch}} \cdot M_{\text{mult}} \cdot M_{\text{switch}}^3
= M_{\text{mult}}^{-1}
= \begin{pmatrix}
\frac{1}{2} & 0 \\ 0 & 2
\end{pmatrix},
\]
which can be realized by a constant sequence of operations. We first add $1$ to $y$ if $b'_0=1$, then apply the division, add $b'_1$ to $y$, and finally apply another division; see Figure~\ref{f:computationtransition}.

For a {\em head-moving transition} $[p, \emptysymbol, d, q]$, if $d=0$, we simply connect state node $p$ to $q$.

If $d=-1$, we move the head to the left. To do so, we extract the symbol $s_{-1}$ from the left side of the tape. We first apply a switch operation exchanging $\theta$ and $y$, and then construct a decision tree as above to determine $v(s_{-1}) = b_0 + 2b_1$ and remove it from $y$. We then apply the inverse switch operation, which can be realized as
\[
M_{\text{switch}}^{-1} = M_{\text{switch}}^3 =
\begin{pmatrix}
0 & -1 \\ 1 & 0
\end{pmatrix}.
\]
Finally, we add the extracted bits to $y$; see Figure~\ref{f:leftheadmovement}.

The remaining case $d=1$ corresponds to a right move of the head and is obtained by inverting the construction for the left move; see Figure~\ref{f:rightheadmovement}.

This completes the construction and proves Theorem~\ref{t:2LRC-Complete}, since we simulate an RTM by a reduced 2-variable linear reversible flowchart.

\begin{figure}
\begin{minipage}[b]{0.48\textwidth}
\includegraphics[width=\textwidth]{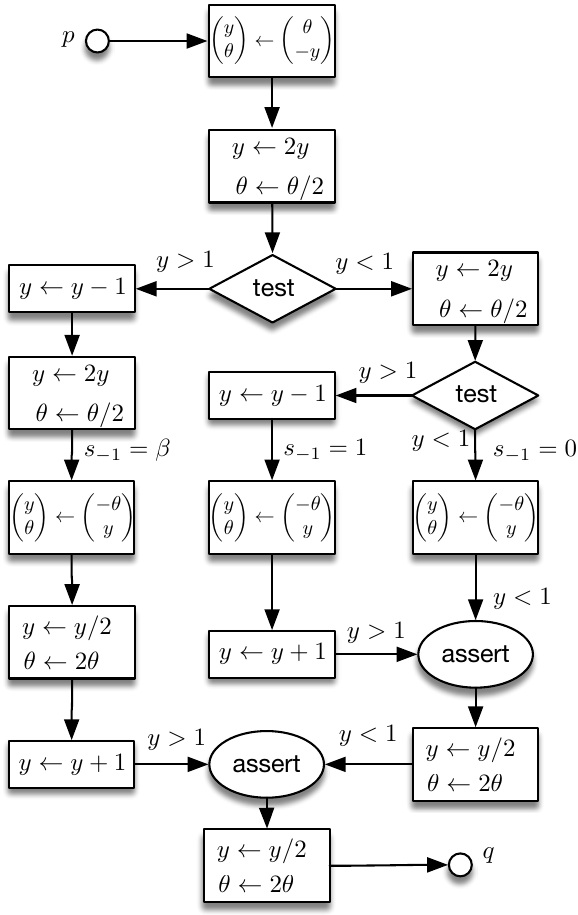}
\caption{Left head movement $[p,\emptysymbol,-1,q]$\label{f:leftheadmovement}}
\end{minipage}
\hfill
\begin{minipage}[b]{0.49\textwidth}
\includegraphics[width=\textwidth]{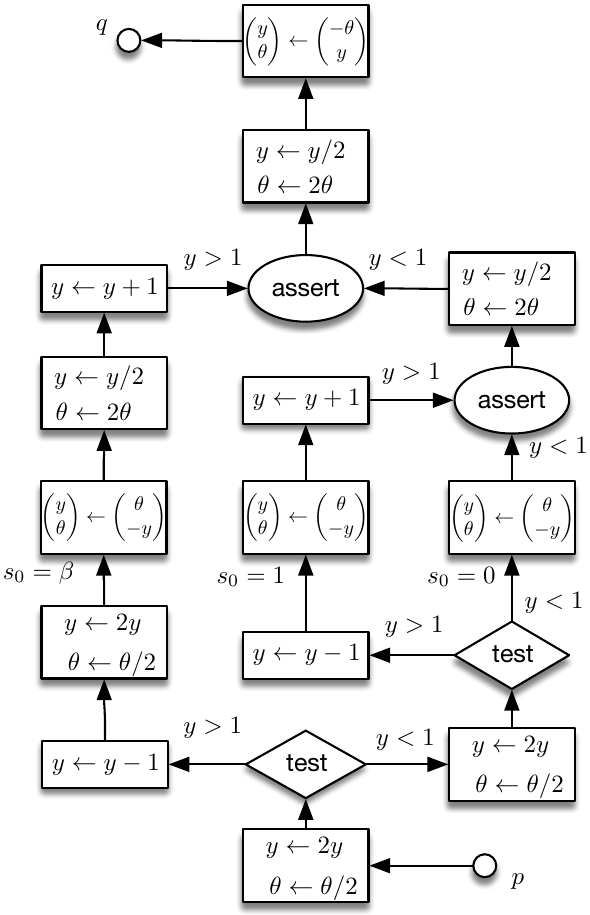}
\caption{Right head movement $[p,\emptysymbol,+1,q]$\label{f:rightheadmovement}}
\end{minipage}
\end{figure}

	%
	%
	%
	%
	%
	%





\section{The ABCD Ray Transfer Model}\label{abcd}

We give a brief introduction to the relevant aspects of the ABCD model in optics, which is used to describe lenses and curved mirrors in two dimensions. For a more detailed treatment, we refer to standard optics textbooks such as~\cite{peatross2015physics}. 

In this model, a reference $x$-axis is fixed along which optical elements, such as lenses or curved mirrors, are placed. At a position $x_0$ on this axis, a light ray is characterized by its offset $y_0$ (the perpendicular distance from the axis) and its angle $\theta_0$ relative to the axis. This state is represented by the vector $(y_0,\theta_0)$.
As the ray propagates through the system, it is refracted by lenses, reflected by mirrors, and travels along straight lines between these elements. At a later position $x_1$, we observe the exiting ray and measure its offset and angle, denoted by $(y_1,\theta_1)$, again with respect to the $x$-axis.

Let us first consider propagation through a region of free space of length $d$, so that $x_1 = x_0 + d$. In this case, the angle remains unchanged, i.e., $\theta_1 = \theta_0$, while the offset satisfies $y_1 = y_0 + d \sin \theta_0$. 
The ABCD model assumes small angles, allowing the approximation $\sin \theta_0 \approx \theta_0$. Using this approximation, propagation through free space can be described by a \emph{space ABCD matrix} $S_d$, which acts via matrix multiplication; see Figure~\ref{f:only-lenses}.
\begin{equation}
\begin{pmatrix} y_1 \\ \theta_1 \end{pmatrix}
=
\underbrace{\begin{pmatrix} 1 & d \\ 0 & 1 \end{pmatrix}}_{S_d}
\cdot \begin{pmatrix} y_0 \\ \theta_0 \end{pmatrix}
\end{equation}

 \begin{figure}[htb]
 \centering
 \begin{minipage}[b]{0.4\textwidth}
	\begin{center}
		\scalebox{.19}{
			\input{fig/LensScapes_ThinLens_01.pspdftex}
		}
	\end{center}
	\caption{The geometry of a thin lens. Note that we assume $d\rightarrow 0$.		\label{fig:ThinLens_02}}
\end{minipage} 
\begin{minipage}[b]{0.59\textwidth}
	\begin{center}
	\scalebox{.78}{
			\includegraphics[width=1\linewidth]{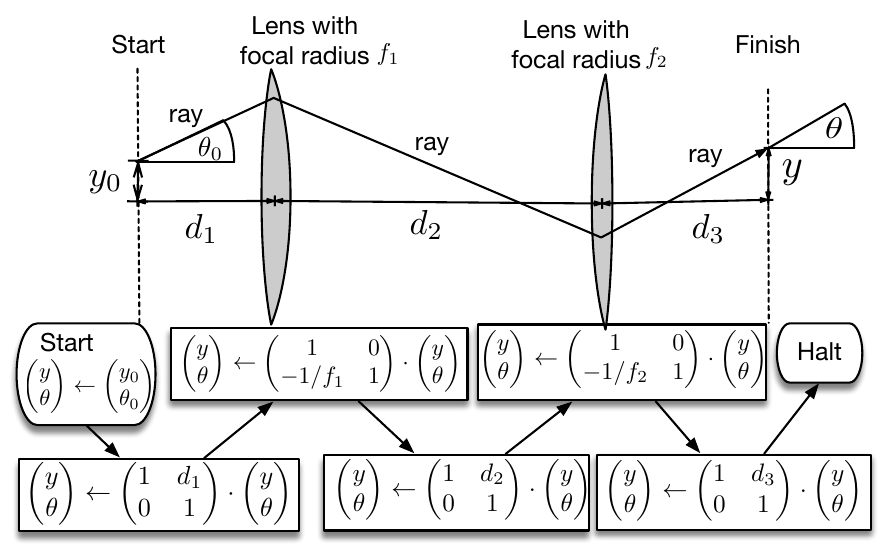}
		}
	\end{center}
\caption{The computation induced by two lenses with focal points $f_1$ and $f_2$  as flowchart. \label{f:only-lenses}}
\end{minipage} 
\end{figure}

Without going into further detail similar approximating matrices can be derived for thin lenses, thick lenses, concave and convex mirrors, which all can be described by Equation~(\ref{eq_matrix}),
	where $A,B,C,D \in \Real$. 
In all cases, the resulting matrices have unit determinant, i.e.,
\begin{equation}
	\begin{vmatrix} A & B \\ C & D \end{vmatrix} = 1 \ .
\end{equation}
Thus, in addition to the space matrix $S_d$, the \emph{thin lens ABCD matrix} $T_f$ is relevant for our considerations.
%
\begin{equation}T_f = \begin{pmatrix} 1 & 0 \\ -\frac{1}{f} & 1 \end{pmatrix}  \qquad \text{where} \qquad \frac{1}{f} = (n-1)\cdot\left(\frac{1}{R_1}-\frac{1}{R_2}\right)\ .
\end{equation}
The value $f\in \Real\setminus\{0\}$ describes the {\it focal length} according to the {\it lens maker's formula}. We assume that the diameter of this lens can be neglected. Here, $n$ describes the refractive index of the glass and we assume that the refractive of air to be $1$. Note that any value $f \in \Real \setminus\{0\}$ can be achieved by an appropriate choice of $R_1$ and $R_2$, where $R_1$ and $R_2$ are the radii of the two sides of the lens, respectively, see Figure~\ref{fig:ThinLens_02}.
 We first investigate, whether all real-valued unit determinant matrices can be the result of ray propagating through thin lenses and space.
It is well known \cite{peatross2015physics} that 
  for all ABCD matrices with $C \neq 0$
there exist  {\em principal planes} in distances $p_1$ before and $p_2$ after it, which reduce it to a thin lens:
	\begin{equation}
 S_{p_2} \cdot \begin{pmatrix} A & B \\ C & D \end{pmatrix} \cdot S_{p_1} =   T_f 
\qquad \text{for}\    f= \frac{-1}{C}, \quad  p_1=\frac{1-D}C, \quad  p_2 = \frac{1-A}C \ .  \end{equation}
Clearly this equation makes only sense if $p_1, p_2 >0$ which is not the case for all values of $A,B,C,D$.

So, the following theorem answers the question, 
whether it is possible to construct every unit determinant matrix with a constant number of single thin lenses.
\begin{theorem} \label{3universalABCD}
	Every rational-valued  matrix with unit determinant can be constructed with exactly three thin lenses of appropriately chosen focal lengths and using appropriate spacing. There are matrices where less than three lenses are not enough.
\end{theorem}
%
%
In the proof, presented in the Appendix, we have 
adapted the focal lengths to the given spaces
,while in reality, the type of lenses is given and one adjusts the spacing. So, we wonder, how many {\em different types} of lenses are necessary to construct all ABCD matrices. In fact a single type of lens with given focal length $f$ is enough.

Clearly for any focal length $f<0$ not all matrices can be generated from such a concave lens, since its corresponding ABCD matrix and the space matrix $S_d$ have only non-negative entries. Therefore, no matrix with negative entries, like $M_{f'}$ with $f'>0$,  can be produced by matrix multiplication. 

But for  focal length $f>0$, any such single thin lens is universal.
\begin{theorem}\label{universallense}
Every rational-valued matrix with determinant 1 can be produced with a constant number of thin lenses of a given rational focal length $f>0$ and adjusted distances.
\end{theorem}
%
%
%
The detailed proof is given in the Appendix. The number of necessary thin lenses is at most 18. We are aware that a construction with a smaller number of thin lenses exists. The exact upper and lower bound for this number is left for future work. The following theorem shows that a constant number of lenses also can fill a large enough distance $d$.

\begin{theorem}\label{anylength}
Given a single type of thin lenses with focal length $f>0$, then for every matrix with determinant~1 there is some $d_0$ such that for all $d\geq d_0$ there is a construction with constant number of lenses which length from start to finish is exactly $d$.
 \end{theorem} 
\FloatBarrier

The proof of this Theorem can be found in the Appendix.

\section{From ABCD Matrices to 2D Ray Tracing Turing Completeness}\label{lens-tracing}

We now construct an optical system based on ABCD lenses and mirrors that simulates an RTM. We start from the reduced 2-variable linear reversible flowchart described in Section~\ref{sec:Red-2-LRFC}. The ray is guided through a grid with cell size $g$; see Figure~\ref{overview-grid}. The ray starts at a grid node, and the Halt node is also represented by a grid node. Grid nodes are either left empty or equipped with mirrors tilted at $45^\circ$.
Empty grid nodes are used for the Start and Halt nodes, as well as for perpendicular crossings. The mirrors redirect an incoming ray with angle $\theta$ to an outgoing ray with angle $\pi/2 - \theta$, either to the left or to the right. We also use mirrors to implement Test and Assertion operations, as shown in Figure~\ref{node-ops}.

\begin{figure}
\begin{minipage}[b]{0.45\textwidth}
\includegraphics[width=1\textwidth]{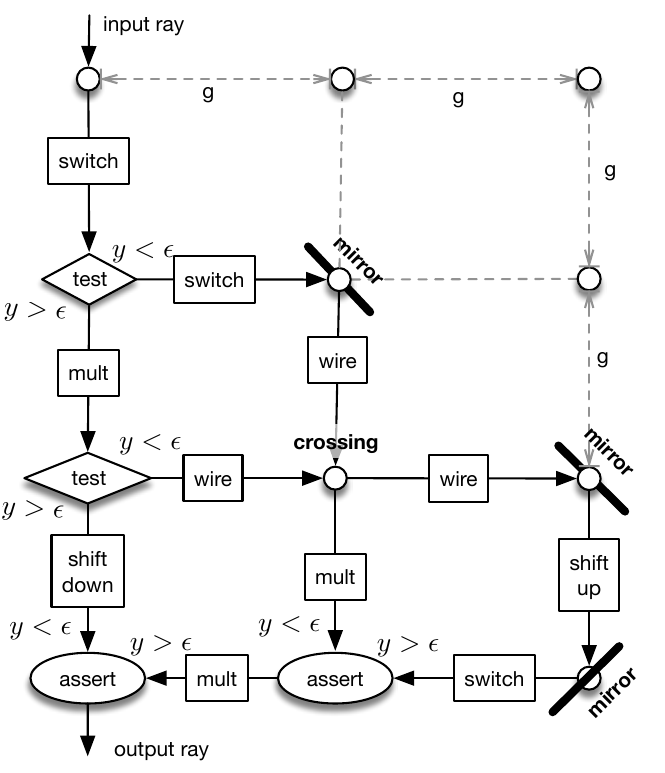}
\caption{Geometry of the grid-based layout for the 2D ray tracing problem. \label{overview-grid}}
\end{minipage}\hfill
\begin{minipage}[b]{0.54\textwidth}
\includegraphics[width=1\textwidth]{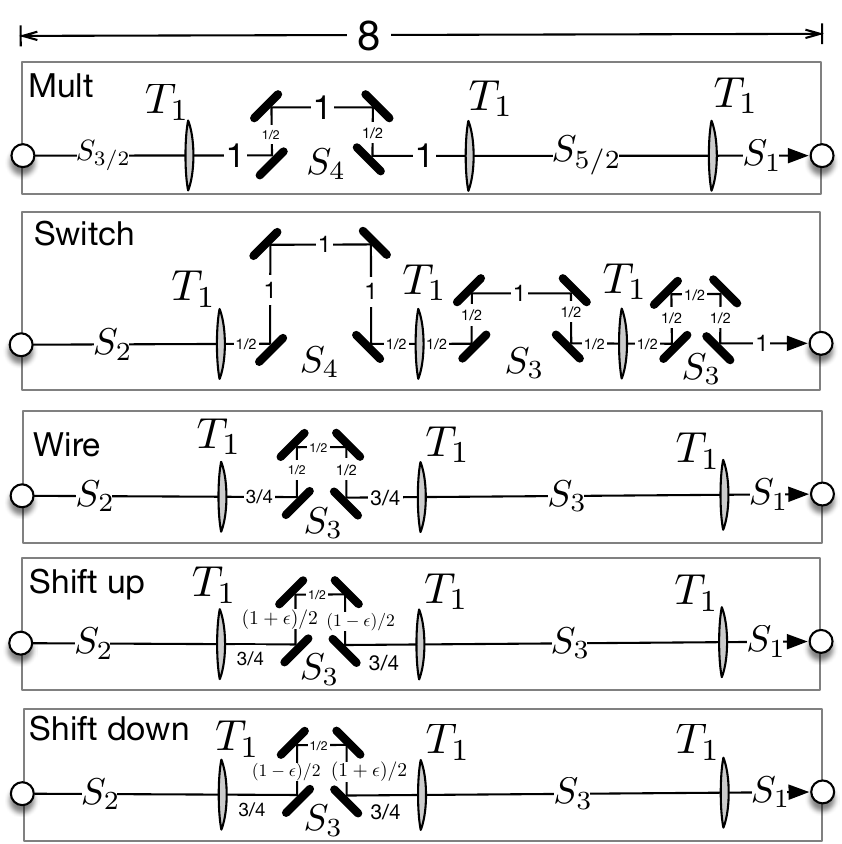}
\caption{Step operations implemented along the edges of the grid. \label{edge-ops}}
\end{minipage}
\end{figure}
Edges connecting grid nodes implement the Step operations \emph{shift-up}, \emph{shift-down}, \emph{multiplication}, \emph{switch}, or a \emph{wire}. The wire corresponds to the identity function, which cannot be realized by empty space alone, since this would result in multiplication by the space matrix $S_g$. We use four rotated versions of each component to accommodate rays traveling in the four cardinal directions.

The operations of the reversible flowchart can be implemented using mirrors and lenses modeled by ABCD matrices as follows. We begin with the Step operations, which must match the grid spacing. In Theorem~\ref{anylength}, we showed that for any fixed lens with focal length $f>0$, one can choose a grid size $g$ such that a construction of exactly this length exists. However, this general result cannot be applied directly here, since the resulting lens positions may be irrational, as they arise from solving quadratic equations.
%
%
%
%
%
%
%
%
%
%
Since for the hardness result it suffices to pick one instance, we choose (among many possibilities) the following combinations 
resulting in three components of overall lengths $9$ and $12$.
\begin{eqnarray} \label{eqn_matrixeg}
	M_\text{mult} & = & S_{1}\cdot T_{1}\cdot S_{5/2}\cdot T_{1}\cdot S_{4}\cdot T_{1}\cdot S_{3/2} \label{mult} \\
	M_\text{switch} & = &  S_{3}\cdot T_{1}\cdot S_{3}\cdot T_{1}\cdot S_{4}\cdot T_{1}\cdot S_{2} \\
	M_\text{wire} & = &  S_{1}\cdot T_{1}\cdot S_{3}\cdot T_{1}\cdot S_{3}\cdot T_{1}\cdot S_{2} \label{switch} 
	\label{wire} 
 \ .
\end{eqnarray}
A remaining issue arises from the choice of $y,\theta \in (-2,2)$, which may cause the edges to become too wide. As a consequence, lenses may interfere with other parallel edges. Moreover, the underlying assumption of the ABCD model, namely $\sin \theta \approx \theta$, is clearly violated.
We resolve this problem by scaling the interval to $(-2\epsilon,2\epsilon)$ for some rational $\epsilon > 0$, for example $\epsilon=\frac{1}{100}$, and replacing the tape encoding by
\(
y = \epsilon \sum_{i=0}^{\infty} v(s_i)\cdot 4^{-i-1}, \) and \(
\theta = \epsilon \sum_{i=1}^{\infty} v(s_{-i})\cdot 4^{-i}.
\)
Accordingly, we replace all Test and Assert conditions by $y < \epsilon$, and change the shift operations to $y \leftarrow y + \epsilon$ (\emph{shift-up}) and $y \leftarrow y - \epsilon$ (\emph{shift-down}), without changing anything else. 

In order to obtain a smaller grid size  $g=8$, we fold all three mirror lens systems using a periscope mirror construction; see Figure~\ref{edge-ops}. 
This periscope is adapted to transform the wire into \emph{shift-up} and \emph{shift-down} operations, which moves the ray respectively up or down by $\epsilon$.
This yields the edge operations shown in Figure~\ref{edge-ops} and the node operations in Figure~\ref{node-ops}. Note that the matrix multiplication operations remain unchanged, since the transformation corresponds to a uniform scaling of all values.

\begin{figure}
\includegraphics[width=0.9\textwidth]{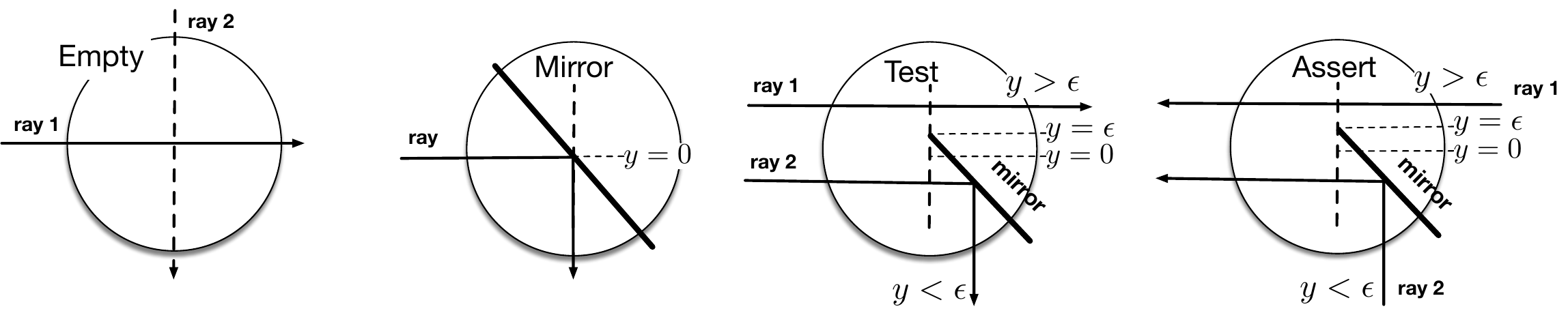}
\caption{Crossing/connection, redirection, Test, and  Assertion at the nodes of the ray-tracing grid. \label{node-ops}}
\end{figure}

Summarizing, to prove the hardness of the 2D ray tracing problem, we start from a given reversible Turing machine and construct the corresponding reduced 2-variable linear reversible flowchart as described in Section~\ref{sec:Red-2-LRFC}. We then embed this flowchart into a grid of size $g$, placing Start, Halt, Test, and Assert nodes at grid points such that all connections can be routed along separate edges. Crossings and arbitrarily long wires are realized using empty nodes.
Next, we apply the scaling with a rational $\epsilon>0$, e.g., $\epsilon = 1/100$, and replace each grid node by the corresponding mirror component for  Test, Assertion or redirection. Finally, we substitute each edge operation by the corresponding three-lens construction as described in Equations~(\ref{mult})--(\ref{wire}). The correctness of the construction follows from the implementations shown in Figure~\ref{edge-ops} and Figure~\ref{node-ops}.

Hence, given an input ray with offset $y$ encoding the initial tape content and an empty tape encoded in $\theta$, the output ray encodes the final tape configuration via its offset and angle, provided that the RTM halts. Moreover, the constructed system reaches the Halt node if and only if the simulated RTM halts. This establishes the main result: two-dimensional ray tracing with ABCD-modeled lenses and mirrors is Turing-complete.

\FloatBarrier



\section{Conclusions and Open Problems}
\label{conclusions}

We have shown that the two-dimensional ray tracing problem with thin lenses and plane mirrors is Turing-complete, thereby resolving the question of whether full three-dimensional space is necessary for computational universality in optical systems. Our construction demonstrates that the classical ABCD matrix formalism from paraxial optics provides a natural and intuitive framework for universal computation. 
The main tool in our simulation of a reversible Turing machine is the close relationship between ray tracing and the concept of reversible flowcharts. We show that such flowcharts are Turing-complete even when restricted to two rational variables and linear functions. In combination with a detailed analysis of the capabilities of ABCD-modeled lenses, this yields our main result: the Turing-completeness of the two-dimensional ray tracing problem with ABCD-modeled lenses and mirrors.

%
%


Unlike the ray particle tracing model introduced in~\cite{adejoh_fsttcs_final}, which simulates a two-stack automaton and encodes one of the stacks in the arrival time (thereby introducing a 2.5-dimensional construction), our model operates entirely within the two-dimensional plane. 
Furthermore, our construction uses only thin lenses and plane mirrors, i.e., stationary components. In particular, we eliminate the need for moving walls (used for time-offset manipulation), one-way gates, or parabolic mirrors. It is worth noting that the constructions in~\cite{adejoh_fsttcs_final} could be simplified by simulating a reversible Turing machine. In that case, one-way gates would no longer be required.

As an open problem, we ask how many lenses of a given focal length $f>0$ are necessary to realize an arbitrary matrix with unit determinant. Clearly, our bound of 18 lenses can be improved. It also remains open whether a lens system with rational parameters can be constructed for sufficiently large rational total length 
(from start to finish). Moreover, a construction using concave and convex mirrors in place of lenses appears feasible within the ABCD model. Finally, the ultimate question is whether Turing-completeness of the two dimensional ray tracing problem can be established for natural optical systems that do not rely on approximations.


\newpage
\bibliography{lit_for_illumination-mirror_tools_for_time}


\newpage
\appendix

%

\section{Appendix: Combination of Thin Lenses for Generating ABCD Matrices}
\label{combo_lenses}

The following proposition is used in the lower bound argument showing that fewer than three lenses are insufficient to construct all unit determinant matrices.

\begin{proposition}\label{proponelense}
A rational-valued $2 \times 2$ matrix $M$ with unit determinant can be realized using a single thin lens with focal length $f$, i.e.,
\[
M = S_{d_2} \cdot T_f \cdot S_{d_1},
\]
if and only if either $(C > 0 \text{ and } A,D \geq 1)$ or $(C < 0 \text{ and } A,D \leq 1)$.
\end{proposition}

\begin{proof}
Let 
\[
M = \begin{pmatrix} A & B \\ C & D \end{pmatrix}
\]
be an ABCD matrix with $\det(M)=1$. We seek parameters $d_1, d_2 \in \mathbb{R}_{\ge 0}$ and $f \in \mathbb{R}\setminus\{0\}$ such that
\[
M = S_{d_2} \cdot T_f \cdot S_{d_1}.
\]
A direct computation yields
\[
S_{d_2} \cdot T_f \cdot S_{d_1}
=
\begin{pmatrix}
1 - \frac{d_2}{f} & d_1 \left(1 - \frac{d_2}{f}\right) + d_2 \\
-\frac{1}{f} & 1 - \frac{d_1}{f}
\end{pmatrix}.
\]
Matching coefficients, we obtain (assuming $C \neq 0$):
\[
f = -\frac{1}{C}, \quad
d_1 = \frac{D - 1}{C}, \quad
d_2 = \frac{A - 1}{C}.
\]

Since physical distances must satisfy $d_1, d_2 \ge 0$, we obtain the following conditions:
\begin{itemize}
    \item If $C > 0$, then $d_1 \ge 0$ and $d_2 \ge 0$ imply $A \ge 1$ and $D \ge 1$.
    \item If $C < 0$, then $d_1 \ge 0$ and $d_2 \ge 0$ imply $A \le 1$ and $D \le 1$.
\end{itemize}
This proves the claim.
\end{proof}

For the case $C = 0$, a construction using at most one lens can only realize free-space propagation matrices of the form
\[
S_B = \begin{pmatrix} 1 & B \\ 0 & 1 \end{pmatrix},
\]
which implies $A , D = 1$ and $B \ge 0$.
\begin{theorem}\label{twolenses}
Every rational-valued $2 \times 2$ matrix with unit determinant and $C,D \neq 0$ can be constructed using two thin lenses with appropriate spacing if and only if one of the following conditions holds:
\begin{enumerate}
    \item $B = C = 0$ and $A < 0$,
    \item $C = 0$, $B \neq 0$, and ($A < 0$ or $B > 0$),
    \item $B = 0$, $C \neq 0$, and ($A < 0$ or $C > 0$),
    \item $B, C \neq 0$, and ($A < 0$ or $B > 0$ or $C > 0$ or $D < 0$).
\end{enumerate}
\end{theorem}

\begin{proof}
Let
\[
M = \begin{pmatrix} A & B \\ C & D \end{pmatrix}, 
\qquad \det(M) = 1.
\]
We seek parameters $d_1, d_2, d_3 \ge 0$ and $f_1, f_2 \in \mathbb{R} \setminus \{0\}$ such that
\[
M = S_{d_3} \cdot T_{f_2} \cdot S_{d_2} \cdot T_{f_1} \cdot S_{d_1}.
\]
A direct computation yields
\begin{equation}\label{2leq}
\begin{pmatrix} A & B \\ C & D \end{pmatrix}
=
\begin{pmatrix}
1 - \frac{d_3}{f_2} - \frac{d_2 \left(1 - \frac{d_3}{f_2}\right) + d_3}{f_1}
&
d_1 \!\left(1 - \frac{d_3}{f_2} - \frac{d_2 \left(1 - \frac{d_3}{f_2}\right) + d_3}{f_1}\right)
+ d_2 \left(1 - \frac{d_3}{f_2}\right) + d_3
\\[1ex]
-\frac{1}{f_2} - \frac{1 - \frac{d_2}{f_2}}{f_1}
&
d_1 \!\left(-\frac{1}{f_2} - \frac{1 - \frac{d_2}{f_2}}{f_1}\right)
+ 1 - \frac{d_2}{f_2}
\end{pmatrix}.
\end{equation}

We distinguish cases.

\begin{enumerate}
\item \textbf{$B = C = 0$.}
Then $A D = 1$ and $A,D \neq 0$. From
\[
0 = -\frac{1 - \frac{d_2}{f_2}}{f_1} - \frac{1}{f_2}
\]
we obtain $f_2 = d_2 - f_1$.

For $A = 1$ this implies $d_2 = 0$, hence no separation of the lenses; thus $A = D = 1$ is impossible.

Solving for $f_1, f_2, d_3$ gives
\begin{align}
f_1 &= -\frac{d_2}{A - 1}, &
f_2 &= \frac{A d_2}{A - 1}, &
d_3 &= -A^2 d_1 - A d_2.
\end{align}

For $A > 0$, no choice of $d_1 \ge 0$, $d_2 > 0$ yields $d_3 \ge 0$.
For $A < 0$, choosing $d_2 \ge (-A)d_1$ ensures $d_3 \ge 0$.

\item \textbf{$C = 0$, $B \neq 0$.}
Then $A D = 1$ and again $f_1 = d_2 - f_2$, implying $f_1 = -\frac{f_2}{A}$.
Solving yields
\begin{align}
f_2 &= \frac{A d_2}{A - 1}, &
f_1 &= \frac{d_2}{1 - A}.
\end{align}
Substituting into the expression for $B$ gives
\begin{equation}
d_2 = B - A d_1 - \frac{d_3}{A}.
\end{equation}
We require $d_1, d_3 \ge 0$ and $d_2 > 0$.
This is impossible if $A > 0$ and $B < 0$.
Otherwise (i.e., $A < 0$ or $B > 0$), a feasible choice exists.

\item \textbf{$B = 0$, $C \neq 0$.}
Solving~\eqref{2leq} for $f_1, f_2, d_2$ yields
\begin{align}
f_1 &=
\frac{A^2 d_1 + d_3 - A C d_1 d_3}
     {A (A - 1 - C d_3)}, \label{achtzehn} \\[0.5ex]
f_2 &=
\frac{A C d_1 d_3 - A^2 d_1 - d_3}
     {A - 1 + A C d_1}, \label{noinzaehn} \\[0.5ex]
d_2 &= C d_1 d_3 - A d_1 - \frac{d_3}{A}.
\end{align}
We require $d_1, d_3 \ge 0$ and $d_2 > 0$.
This is impossible if $C < 0$ and $A > 0$.
Otherwise ($A < 0$ or $C > 0$), feasible choices exist, and the denominators in
\eqref{achtzehn}–\eqref{noinzaehn} can be kept nonzero.

\item \textbf{$B, C \neq 0$.}
Solving~\eqref{2leq} for $f_1, f_2, d_2$ gives
\begin{align}
f_1 &=
\frac{-A C d_1 + A D + C^2 d_1 d_3 - C D d_3 - 1}
     {C (-A + C d_3 + 1)}, \label{zwoeins} \\[0.5ex]
f_2 &=
\frac{-A C d_1 + A D + C^2 d_1 d_3 - C D d_3 - 1}
     {C (C d_1 - D + 1)}, \label{zwozwo} \\[0.5ex]
d_2 &=
C d_1 d_3 - D d_3 - A d_1 + \frac{A D - 1}{C}
= C d_1 d_3 - D d_3 - A d_1 + B.
\end{align}
If $A, D \ge 0$ and $B, C < 0$, then no choice of $d_1, d_3 \ge 0$ yields $d_2 \ge 0$.
In all other cases, a feasible solution exists, and denominators in
\eqref{zwoeins}–\eqref{zwozwo} can be kept nonzero.
\end{enumerate}
\end{proof}
\begin{theorem}[Theorem~\ref{3universalABCD}]
	Every rational-valued  matrix with unit determinant can be constructed with exactly three thin lenses of appropriately chosen focal lengths and using appropriate spacing. There are matrices where less than three lenses are not enough.
	\end{theorem}
\begin{proof}
	The lower bound follows by Proposition~\ref{proponelense} and Theorem~\ref{twolenses}.

	Let $M$ with $|M|=1$ be the potential  ABCD matrix, then we have to find
	$d_1, \ldots, d_4>0$ and $f_1, f_2, f_3\neq 0$ such that
	\begin{eqnarray} \label{3leq}
	M&=&  \begin{pmatrix} A & B \\ C & D \end{pmatrix}
	= S_{d_4} \cdot T_{f_3} \cdot S_{d_3} \cdot T_{f_2}
	\cdot S_{d_2} \cdot T_{f_1} \cdot S_{d_1}\quad  \text{which yields to}\\
A &=&
1
- \frac{d_4 + d_3\!\left(1 - \frac{d_4}{f_3}\right)}{f_2} \nonumber
\\
&&
- \frac{
  d_4 + d_3\!\left(1 - \frac{d_4}{f_3}\right)
  + d_2\!\left(
    1 - \frac{d_4 + d_3\!\left(1 - \frac{d_4}{f_3}\right)}{f_2}
    - \frac{d_4}{f_3}
  \right)
}{f_1}
- \frac{d_4}{f_3},  \label{3A}
\\[0.6em]
B &=&
d_4
+ d_3\!\left(1 - \frac{d_4}{f_3}\right)
+ d_2\!\left(
  1 - \frac{d_4 + d_3\!\left(1 - \frac{d_4}{f_3}\right)}{f_2}
  - \frac{d_4}{f_3}
\right) \nonumber
\\
&&
+ d_1\!\left(
  1 - \frac{d_4 + d_3\!\left(1 - \frac{d_4}{f_3}\right)}{f_2}
\right. \nonumber
\\
&&\left.
  - \frac{
    d_4 + d_3\!\left(1 - \frac{d_4}{f_3}\right)
    + d_2\!\left(
      1 - \frac{d_4 + d_3\!\left(1 - \frac{d_4}{f_3}\right)}{f_2}
      - \frac{d_4}{f_3}
    \right)
  }{f_1}
  - \frac{d_4}{f_3} \label{3B}
\right). \\
C &=&
- \frac{1 - d_3/f_3}{f_2}
- \frac{
  1
  + d_2\!\left(
    - \frac{1 - d_3/f_3}{f_2}
    - \frac{1}{f_3}
  \right)
  - \frac{d_3}{f_3}
}{f_1}
- \frac{1}{f_3},\label{3C}
\\[0.6em]
D &=&
1
+ d_2\!\left(
  - \frac{1 - d_3/f_3}{f_2}
  - \frac{1}{f_3}
\right) \nonumber
\\
&&\quad
+ d_1\!\left(
  - \frac{1 - d_3/f_3}{f_2}
  - \frac{
    1
    + d_2\!\left(
      - \frac{1 - d_3/f_3}{f_2}
      - \frac{1}{f_3}
    \right)
    - \frac{d_3}{f_3}
  }{f_1}
  - \frac{1}{f_3}
\right)
- \frac{d_3}{f_3}. \label{3D}
\end{eqnarray}

	First consider the case $A=0$, $B\neq 0$. This gives the following solutions if we solve Equations~(\ref{3A}), (\ref{3B}), (\ref{3D}) for $f_1, f_2,f_3$:
	\begin{eqnarray}
		f_1 &=& -\frac{
			d_2 \left(B^2 - B D d_4 - d_1 d_4\right)}
		{-B^2 + B d_3 + B D d_4 + d_1 d_4 + d_2 d_4},
		\\
		f_2 &=& \frac{
			B d_2 d_3}
		{-B^2 + B d_2 + B d_3 + B D d_4 + d_1 d_4},
		\\
		f_3 &=& -\frac{
			d_3 \left(B^2 - B D d_4 - d_1 d_4\right)}
		{-B^2 + B d_2 + B D d_3 + d_1 d_3 + B D d_4 + d_1 d_4}
	\end{eqnarray}
	Note that positive choices for $d_1,d_2,d_3,d_4$ exist such that the numerator and denominators are non-zero, which is enough to settle this case.
	
	Similarly, we get for $A=0$ and $C \neq 0$, if we solve   Equations~(\ref{3A}), (\ref{3C}), (\ref{3D}) for $f_1, f_2,f_3$:
	
	\begin{eqnarray}
		f_1 &=& \frac{
			d_2 \left(-1 - C D d_4 + C^2 d_1 d_4\right)}
		{-1 - C d_3 - C D d_4 + C^2 d_1 d_4 + C^2 d_2 d_4}\quad ,
		\\
		f_2 &=& \frac{
			C d_2 d_3}
		{1 + C d_2 + C d_3 + C D d_4 - C^2 d_1 d_4}\quad ,
		\\
		f_3 &=& \frac{
			d_3 \left(-1 - C D d_4 + C^2 d_1 d_4\right)}
		{-1 - C d_2 - C D d_3 + C^2 d_1 d_3 - C D d_4 + C^2 d_1 d_4}\quad .
	\end{eqnarray}
	which allows positive choices for  $d_1,d_2,d_3,d_4$ such that $f_1, f_2, f_3$ are well-defined and non-zero.

	For the last case $A\neq 0$ we solve the Equations~(\ref{3A}), (\ref{3B}), (\ref{3C}) for $f_1,f_2,f_3$ resulting in the following solutions.
	\begin{eqnarray}
		f_1&=&
		\frac{d_2\left(-A B + A^2 d_1 + d_4 + B C d_4 - A C d_1 d_4\right)}
		{-A B + A^2 d_1 + A^2 d_2 + A d_3 + d_4 + B C d_4 - A C d_1 d_4 - A C d_2 d_4}\\
		f_2&=&\frac{A d_2 d_3}
		{-A B + A^2 d_1 + A d_2 + A d_3 + d_4 + B C d_4 - A C d_1 d_4}
		\\
		f_3&=&\frac{-A B d_3 + A^2 d_1 d_3 + d_3 d_4 + B C d_3 d_4 - A C d_1 d_3 d_4}
		{-A B + A^2 d_1 + A d_2 + d_3 + B C d_3 - A C d_1 d_3 + d_4 + B C d_4 - A C d_1 d_4}
	\end{eqnarray}
	
	Again, we have to check that the numerators and denominators on the right side are never $0$. 
	which can be done for appropriate positive choices of $d_2, d_3, d_4$ given any $d_1>0$.
\end{proof}
\begin{theorem}[Theorem~\ref{universallense}]
Every rational-valued matrix with unit determinant can be produced with 18 thin lenses of a given rational focal length $f>0$ and rational valued positive distances.
\end{theorem}
\begin{proof}
In Theorem~\ref{3universalABCD} we showed that every rational-valued unit determinant matrix  can  be produced with three thin lenses $f_1, f_2, f_3 \neq 0$. So, we can concentrate on producing each of these thin lenses with focal length $f'\in \{f_1, f_2, f_3\}$ using a given thin lenses with focal length $f>0$.

Observe that we can produce $-T_{f'}$ by the following construction
\begin{equation}\label{trickytrick}
\begin{pmatrix}  -1 & 0 \\ 1/f' & -1 \end{pmatrix} = S_f \cdot T_f \cdot S_{\frac{f(f+2f')}{f'}} \cdot T_f \cdot S_f
\end{equation}
if the space in the middle between the lenses is positive, i.e. $\frac{f(f+2f')}{f'} > 0$. This is the case for $f' \in \Ratio \setminus [-f/2, 0]$. For the construction of $T_{f'}$ define 
\begin{equation} D_f := S_{f}\cdot T_{f} \cdot S_{f} = \begin{pmatrix}  0 & f \\ -1/f & 0 \end{pmatrix} \end{equation} 
and observe that 
\begin{equation} D_f^2 = \begin{pmatrix}  -1 & 0 \\ 0 & -1 \end{pmatrix}\end{equation}
Therefore, for $f' \in \Ratio \setminus [-f/2, 0]$ we can give a four lenses construction as follows.
\begin{equation}T_{f'}= D_f^2 \cdot S_f \cdot T_f \cdot S_{\frac{f(f+2f')}{f'}} T_f \cdot S_f
\label{nearlyallTfy}\end{equation}
This implies that for $f_1,f_2,f_3 \in \Ratio \setminus [-f/2, 0]$ there is a construction with 12 lenses.

It remains to consider the case $f' \in [-f/2, 0)$. Note that the construction in
(\ref{trickytrick}) can be applied to $-T_{-f'}$ 
\begin{equation}
  -T_{-f'}
  = \begin{pmatrix} -1 & 0 \\ -1/f' & -1 \end{pmatrix}
  = S_f \cdot T_f \cdot S_{\frac{f(f-2f')}{-f'}} \cdot T_f \cdot S_f .
\end{equation}
Thus, we obtain $T_{-f'} = D_f^2 \cdot (-T_{-f'})$ using four lenses.
Since $f' \notin [f/2, 0)$, we can again apply (\ref{trickytrick}) by substituting the four lens construction for $T_{-f'}$ and obtain
\begin{eqnarray}
  -T_{f'}
   &=& S_{-f'} \cdot T_{-f'} \cdot S_{-f'} \cdot T_{-f'} \cdot S_{-f'} \\
  & =& S_{-f'} \cdot (-T_{-f'}) \cdot S_{-f'} \cdot (-T_{-f'}) \cdot S_{-f'} .
\end{eqnarray}
using $\frac{(-f')(-f' + 2f')}{f'} = -f'$. For  $f' \in [-f/2, 0)$ the overall construction is therefore
\begin{equation}
T_{f'} =
 D_f^2 \cdot 
S_{-f'} \cdot \underbrace{ S_f \cdot T_f \cdot S_{\frac{f(f-2f')}{-f'}} \cdot T_f \cdot S_f}_{-T_{-f'}}
\cdot 
 S_{-f'} \cdot   \underbrace{ S_f \cdot T_f \cdot S_{\frac{f(f-2f')}{-f'}} \cdot T_f \cdot S_f}_{-T_{-f'}} \cdot S_{-f'}
 \end{equation}
which is an optical system using 6 thin lenses of focal length $f$ for implementing a focal length $f' \in \{f_1,f_2,f_3\}$. Since Theorem~\ref{3universalABCD} shows that any unit determinant matrix can be expressed by these three lenses, we have shown a construction with at most 18 lenses for the general case.
%
%
\end{proof}

%
%

\begin{theorem}[Theorem \ref{anylength}]
Given a single type of thin lenses with focal length $f>0$, then for every ABCD matrix there is some $d_0$ such that for all $d\geq d_0$ there is a construction with constant number of lenses which length from start to finish is exactly $d$.
 \end{theorem}
 \begin{proof}
 We first construct the ABCD matrix according to the construction given in the proof of Theorem~\ref{universallense} using some space $d'$. 
 Now observe that 
 \begin{equation}
M :=  \begin{pmatrix} 0 & x \\ -1/x & 0  \end{pmatrix}  =
 S_{\frac{1}{x} + f} \cdot T_{f} \cdot S_{f (2+ f x)} \cdot T_{f}  \cdot   S_{\frac{1}{x} + f}
 \end{equation}
 Since $M^4 = M_{\text{wire}}$ we have a variable length construction for connections with overall length $d  =  8/x + 4 f (4 + f x)$.  Solving for $x$ yields
 \[
x=\frac{d-16f\pm\sqrt{d^2-32df+128f^2}}{8f^2}
\quad (f\neq 0)
\]
So, for $ d \ge (16 + 8\sqrt{2})f$ there exists always a real-valued solution for $x$ which can be used for a filling up the gap between the finish line and the universal ABCD matrix construction.  Note that rational valued solutions for $x$ exist, if $d$ and $f$  are chosen appropriately from the rational numbers. 
 \end{proof}










\end{document}